\documentclass[11pt]{article}
\usepackage{times}
\usepackage{xspace,comment,calc}
\usepackage{amssymb,amsmath,amsthm,amsfonts,amstext}
\usepackage[shortlabels]{enumitem}
\usepackage{multirow}
\usepackage{graphicx}
\usepackage{bbm}
\usepackage{comment}
\usepackage{xifthen}

\usepackage[T1]{fontenc}  
\usepackage[b]{esvect}    
\usepackage{tikz}         

\makeatletter
\newlength\xvec@height%
\newlength\xvec@depth%
\newlength\xvec@width%
\newcommand{\xvec}[2][]{%
  \ifmmode%
    \settoheight{\xvec@height}{$#2$}%
    \settodepth{\xvec@depth}{$#2$}%
    \settowidth{\xvec@width}{$#2$}%
  \else%
    \settoheight{\xvec@height}{#2}%
    \settodepth{\xvec@depth}{#2}%
    \settowidth{\xvec@width}{#2}%
  \fi%
  \def\xvec@arg{#1}%
  \def\xvec@dd{:}%
  \def\xvec@d{.}%
  \raisebox{.2ex}{\raisebox{\xvec@height}{\rlap{%
    \kern.05em
    \begin{tikzpicture}[scale=1]
    \pgfsetroundcap
    \draw (.05em,0em)--(\xvec@width-.05em,0em);
    \draw (\xvec@width-.05em,0em)--(\xvec@width-.25em,.15em);
    \draw (\xvec@width-.05em,0em)--(\xvec@width-.25em,-.15em);
    \ifx\xvec@arg\xvec@d%
      \fill(\xvec@width*.45,.5ex) circle (.5pt);%
    \else\ifx\xvec@arg\xvec@dd%
      \fill(\xvec@width*.30,.5ex) circle (.5pt);%
      \fill(\xvec@width*.65,.5ex) circle (.5pt);%
    \fi\fi%
    \end{tikzpicture}%
  }}}%
  #2%
}
\makeatother

\renewcommand{\vec}[1]{\xvec[]{#1}}

\newcommand{\np}{{\em NP}\xspace}
\newcommand{\nphard}{\np-hard\xspace}

\newtheorem{theorem}{Theorem}[section]
\newtheorem{lemma}[theorem]{Lemma}
\newtheorem{claim}[theorem]{Claim}

{\theoremstyle{definition}  }

\newenvironment{proofof}[1]{\begin{proof}[Proof of #1]}{\end{proof}}

\newcommand{\bg}[1]{\medskip\noindent{\it #1}}

\newcommand{\R}{\ensuremath{\mathbb R}}

\newcommand{\I}{\ensuremath{\mathcal I}}

\newcommand{\Pc}{\ensuremath{\mathcal P}}

\newcommand{\OPT}{\ensuremath{\mathsf{OPT}}}

\newcommand{\ceil}[1]{\ensuremath{\left\lceil#1\right\rceil}}

\newcommand{\poly}{\operatorname{poly}}

\newcommand{\junk}[1]{}

\newcommand{\sse}{\subseteq}

\newcommand{\bx}{\ensuremath{\overline{x}}}

\newcommand{\hx}{\ensuremath{\hat x}}

\newcommand{\sg}{\ensuremath{\sigma}}

\newcommand{\ve}{\ensuremath{\varepsilon}}
\newcommand{\gm}{\ensuremath{\gamma}}
\newcommand{\ld}{\ensuremath{\lambda}}
\newcommand{\al}{\ensuremath{\alpha}}

\newcommand{\w}{\ensuremath{\omega}}
\newcommand{\kp}{\ensuremath{\kappa}}

\newcommand{\lb}{\ensuremath{\mathsf{lb}}}

\newcommand{\optl}{\ensuremath{\mathit{opt}_\ell}}
\newcommand{\iopt}{\ensuremath{O^*}}

\newcommand{\alg}{\ensuremath{\mathsf{alg}}}
\newcommand{\est}{\ensuremath{\mathsf{est}}}

\newcommand{\load}{\ensuremath{\mathsf{load}}}
\newcommand{\lvec}{\ensuremath{\vec{\load}}}

\newcommand{\mnp}{q}
\newcommand{\sgr}{\ensuremath{d}}
\newcommand{\hsgr}{\ensuremath{\widehat{\sgr}}}

\newcommand{\inpsize}{\ensuremath{\text{input size}}\xspace}
\newcommand{\vL}[1][]{\ifthenelse{\equal{#1}{}}{\ensuremath{\vec{L}}}{\ensuremath{\vec{L({#1})}}}}
\newcommand{\vP}[1][]{\ifthenelse{\equal{#1}{}}{\ensuremath{\vec{P}}}{\ensuremath{\vec{P({#1})}}}}
\newcommand{\topl}{\ensuremath{\mathsf{Top}\text{-}\ell}\xspace}
\newcommand{\pset}{\ensuremath{\mathsf{POS}}}
\newcommand{\minmax}[1][m]{{#1}in-max\xspace}
\newcommand{\minlb}{\ensuremath{\mathsf{MinNormLB}}\xspace}
\newcommand{\bsg}{\ensuremath{\overline{\sg}}}
\newcommand{\bo}{\ensuremath{0}} 

\title{Simpler and Better Algorithms for Minimum-Norm Load Balancing}
\author{
         Deeparnab Chakrabarty\thanks{{\tt deeparnab@gmail.com}.
         Dept. of Computer Science, Dartmouth College, Hanover, NH 03755-3510, USA.}
\and
         Chaitanya Swamy\thanks{{\tt cswamy@uwaterloo.ca}.
         Dept. of Combinatorics and Optimization, Univ. Waterloo, Waterloo, ON N2L 3G1,
         Canada. 
         Supported in part by NSERC grant 327620-09 and an NSERC Discovery Accelerator
         Supplement Award.}
}
\date{}

\begin{document}

\maketitle

\begin{abstract}
Recently, Chakrabarty and Swamy (STOC 2019) introduced the 
{\em minimum-norm load-balancing} problem on unrelated machines, wherein we are given a
set $J$ of jobs that need to be scheduled on a set of $m$ unrelated machines, and a
monotone, symmetric norm; 
We seek an assignment $\sg:J\mapsto[m]$ that minimizes the norm of the resulting load
vector $\lvec_\sg\in\R_+^m$, where $\lvec_\sg(i)$ is the load on machine $i$ under the
assignment $\sg$. Besides capturing all $\ell_p$ norms, symmetric norms also capture other
norms of interest including top-$\ell$ norms, and ordered norms.
Chakrabarty and Swamy (STOC 2019) give a $(38+\ve)$-approximation algorithm for
this problem via a general framework they develop for minimum-norm optimization that
proceeds by first carefully reducing this problem (in a series of steps) to a problem
called \minmax ordered load balancing, and then devising a so-called deterministic
oblivious LP-rounding algorithm for ordered load balancing.

We give a direct, and simple $4$-approximation algorithm%
\footnote{Since the norm could be irrational, the factor is really $4+\ve$, achieved in
time $\poly\bigl(\text{input size},\log(1/\ve)\bigr)$.}
	for the minimum-norm load balancing based on rounding a (near-optimal) solution to a novel
convex-programming relaxation for the problem. Whereas the natural convex program encoding
minimum-norm load balancing problem has a large non-constant integrality gap, we show that
this issue can be remedied by including a key constraint that bounds the 
``norm of the job-cost vector.''
Our techniques also yield a (essentially) $4$-approximation for:
(a) {\em multi-norm load balancing}, wherein we are given multiple monotone symmetric
norms, and we seek an assignment respecting a given budget for each norm;
(b) the best {\em simultaneous approximation factor} achievable for all symmetric norms
for a given instance. 
\end{abstract}


\section{Introduction} \label{intro}
In the {\em minimum-norm load-balancing} (\minlb) problem, we are given a set $J$ of $n$ jobs,
a set of $m$ machines, and processing times $p_{ij} \ge 0$ for all $i\in[m]$ and 
$j\in J$. We use $[m]$ to denote $\{1,\ldots,m\}$. 
We are also given a monotone, symmetric norm $f:\R^m \to \R_+$.
Recall that by definition of norm, this means that $f$ satisfies:  
(i) $f(x)=0$ iff $x=0$; 
(ii) $f(x+y)\leq f(x)+f(y)$ for all $x,y\in\R^m$ (triangle inequality); and 
(iii) $f(\ld x)=|\ld|f(x)$ for all $x\in\R^m,\ld\in\R$ (homogeneity). 
(Properties (ii) and (iii) imply that $f$ is convex.)
Monotonicity means that $f(x)\leq f(y)$ for all $x,y\in\R^m$ such that 
$x_i(y_i-x_i)\geq 0$ for all $i\in[m]$; 
symmetry means that permuting the coordinates of $x$ does not affect its
norm, i.e., $f(x)=f\bigl(\{x_{\pi(i)}\}_{i\in[m]}\bigr)$ for every $x\in\R^m$ and every 
permutation $\pi:[m]\mapsto[m]$. 

The goal is to find an assignment $\sg: J\to[m]$ that minimizes the norm (under $f$) of
the induced load vector. More precisely, an assignment $\sg$ induces the $m$-dimensional
load vector $\lvec_\sg \in \R^m_+$ where  
$\lvec_\sg(i) := \sum_{j: \sg(j) = i} p_{ij}$. The objective is to find $\sg$ that
minimizes $f(\lvec_\sg)$.

Besides $\ell_p$-norms, monotone symmetric norms capture \topl norms---sum of $\ell$
largest coordinates in absolute value---and ordered
norms (which are a nonnegative linear combination of \topl norms). 
The minimum-norm load-balancing problem was recently introduced by Chakrabarty and
Swamy~\cite{ChakrabartyS19}. They develop a general framework for minimum-norm
optimization problems based on reducing the problem to a special case called min-max
ordered optimization, and devise a so-called deterministic oblivious
rounding~\cite{ChakrabartyS19} to tackle the latter problem, which results in
a $(38+\ve)$-approximation algorithm for \minlb.

Our main result is a simpler $4(1+\ve)$-approximation algorithm for \minlb
that runs in time $\poly\bigl(\text{input size},\log(\frac{1}{\ve})\bigr)$. 

\begin{theorem} \label{4apx}
One can achieve a $4(1+\ve)$-approximation for \minlb in time 
$\poly\bigl(\text{input size},\log(\frac{1}{\ve})\bigr)$, assuming we have a value-oracle  
and subgradient-oracle for the norm $f$. 

More generally, if we have $\w$-approximate value- and
subgradient- oracles for $f$ (see Section~\ref{convsolve}), then one can compute
a $4(1+5\w)(1+\ve)$-approximation to \minlb in time
$\poly\bigl(\text{input size},\log(\frac{1}{\ve})\bigr)$.
\end{theorem}

This is a substantial improvement
over the approximation factor of $38$ obtained in~\cite{ChakrabartyS19}. Moreover, our
algorithm is also simpler and more direct than the one in~\cite{ChakrabartyS19}. 
Notably, our approximation factor is close to the best-known approximation factor (of $2$)
known for the $\ell_\infty$ norm  (wherein \minlb becomes the classical minimum-makespan
problem).  
Our algorithm proceeds by rounding the
solution to a novel convex-programming relaxation of the problem. 
The convex program can be solved (approximately) using an (approximate) first-order
oracle for $f$ that returns the function value, and its subgradient at a given point.

Our techniques also yield a $4(1+\ve)$-approximation for (see Section~\ref{extn}): 
(a) {\em multi-norm load balancing},
wherein we are given multiple monotone, symmetric norms and budgets for each norm, and we
seek an assignment (approximately) respecting these budgets; and
(b) the best {\em simultaneous approximation factor} achievable for all symmetric norms
for a given instance.

\paragraph{Motivation and perspective.} 
One of the reasons for studying \minlb is that it generalizes various load-balancing
problems considered in the literature, and its study therefore yields a unified
methodology for dealing with monotone, symmetric norms.  

Load balancing under the $\ell_\infty$ norm, that is, minimizing the maximum load
(also called the makespan) is a classical scheduling problem that 
has been extensively
studied~\cite{LenstraST90,ShmoysT93,EbenlendrKS14,Svensson12,ChakrabartyKL15,JansenR17}
over the past three decades, both in its full generality for unrelated machines and for
various special cases.  
The best known approximation factor for
the unrelated-machines setting is still $2$~\cite{LenstraST90}, and it is \nphard to
obtain an approximation factor better than $3/2$~\cite{LenstraST90}.
For general $\ell_p$-norms,  
Azar and Epstein~\cite{AzarE05} obtain a $2$-approximation, and 
improved guarantees have been obtained for constant 
$p$~\cite{AzarE05,KumarMPS09,MakarychevS14}.
More recently, the load-balancing problem has also been considered for other monotone,
symmetric norms. \topl- and ordered- norms have been proposed in the location-theory
literature (see ``Other related work'') as a means of interpolating between the $\ell_1$
and $\ell_\infty$ norms (and an alternative to using $\ell_p$ norms), and motivated
by this, 
Chakrabarty and Swamy~\cite{ChakrabartyS19} studied  
the \topl load-balancing problem---minimize the total load on the $\ell$ most loaded
machines---and the ordered load-balancing problem. 
They give a $(2+\ve)$-approximation algorithm in both settings, and also (as noted
earlier) devise a $(38+\ve)$-approximation algorithm for an arbitrary monotone, symmetric
norm.    

For load balancing, there has been considerable interest in 
{\em simultaneous optimization}. Given an instance, the objective is to find an 
assignment that 
{\em simultaneously} approximates a large suite of objective functions. Building upon previous
works~\cite{AlonAWY98,AzarERW04}, Goel and Meyerson~\cite{GoelM06} describe a
$2$-approximation for the problem of simultaneously approximating all monotone symmetric
norms in the {\em restricted assignment} setting. However, it is known that such an
$O(1)$-factor is {\em impossible} in the unrelated-machines setting~\cite{AzarERW04,GoelM06}.
As a byproduct of their \minlb algorithm, in the unrelated-machines setting, Chakrabarty
and Swamy~\cite{ChakrabartyS19} give an {\em instance-wise} 
$(38 + \ve)$-approximation to the best simultaneous approximation-factor possible for the
instance. To elaborate,  let $\alpha^*_\I$ denote the smallest factor for instance $\I$ such
that there exists a schedule that achieves an $\alpha^*_\I$-approximation for all
monotone, symmetric norms;  
the work of \cite{ChakrabartyS19} returns a schedule for $\I$ that achieves a
$38(1+\ve)\alpha^*_\I$-approximation for all monotone, symmetric norms.  
As mentioned above, we devise an algorithm that for every instance $\I$ returns a schedule
that simultaneously achieves a $\bigl(4+O(\ve)\bigr)\alpha^*_\I$-approximation for all
monotone, symmetric norms (see Theorem~\ref{thm:simul-opt}).

\paragraph{Our techniques.}
Since a norm is a convex function, a natural convex-programming relaxation for \minlb is
to minimize the norm of the fractional load vector 
$\vec{L}=\vec{L(x)}:=\bigl\{\sum_j p_{ij}x_{ij}\bigr\}_{i\in[m]}$, where the $x_{ij}$s are
the usual variables denoting if job $j$ is assigned to machine $i$, and we have the usual
job-assignment constraints encoding that every job is assigned to some
machine. This convex program, however, has a large integrality gap, even when $f$ is the
$\ell_\infty$-norm due to the issue that the convex program
could split a large job across multiple machines. 

In the case of the $\ell_\infty$ norm (the makespan minimization problem), the typical way
of circumventing the above issue is to 
``guess'' the optimal value, say $T$, and add constraints to encode that no single job 
contributes more than $T$ to the objective. The usual way of
capturing this is to explicitly set $x_{ij}=0$ 
if $p_{ij}>T$. 
A less common, and weaker, way of encoding this is to enforce that 
$\sum_i p_{ij}x_{ij}\leq T$ for all $j$, that is, the total processing time contribution
of any job $j$ across the 
machines cannot exceed $T$.

For an arbitrary (monotone, symmetric) norm, it is unclear how to extend
either of the above approaches, since the contribution of a job to the objective is a
now a somewhat vague notion.
One way to generalize things would be to encode (either explicitly or in the
alternate weaker sense above) that the ``norm'' of the job-cost vector is at
most $T$, where the job-cost vector is indexed by jobs and the cost for job $j$ (under
$x$) is $P_j:=\sum_i p_{ij}x_{ij}$. But the norm $f$ is defined over $\R^m$, whereas the
job-cost vector lies in $\R^n$. For certain specific (families of) norms---e.g.,
$\ell_p$-norms, top-$\ell$ norm, ordered norm---there is a natural version of the norm
over $\R^n$,%
\footnote{For $\ell_p$-norms, a variant of this that considers the $\ell_p^p$
  expression does work, but this crucially exploits the separability of
  $\ell_p^p$~\cite{AzarE05}.}  
but what does such a constraint mean in general, and how can one encode this? 

Our key insight, which leads to our convex program, is that one can capture the
above consideration by examining the vector ${\vec{P}}\in\R^m$ comprising the
{\em costs of the $m$ most-costly jobs} 
and enforcing the constraint $f(\vec{P})\leq T$; 
since $f$ is monotone, this can be equivalently encoded as
$f\bigl(\{P_j\}_{j\in S}\bigr)\leq T$ for all $S\sse J$ with $|S|=m$.
It is not apparent that such a constraint is valid, but we derive some insights about
symmetric norms and show that this is indeed the case (see Theorem~\ref{thm:valid}). This yields
our convex program \eqref{eq:cp}, which  
can be solved efficiently (within $\ve$ additive error, for any $\ve>0$, in time 
$\poly\bigl(\text{input size},\ln(1/\ve)\bigr)$) using the ellipsoid method provided we
have a value oracle and subgradient oracle for $f$. 

Rounding a solution $x$ to the convex program is now quite easy. Let $\vec{L}\in\R^m$
denote the load-vector arising from $x$. We use a filtering step to ensure
that each job $j$ is only assigned to machines $i$ for which $p_{ij}\leq 2P_j$. This
causes a factor-$2$ blowup in the machine loads. Now we use the rounding algorithm of
Shmoys and Tardos~\cite{ShmoysT93} for the generalized assignment problem (GAP). The
resulting assignment $\sg$ has load-vector at most $2\vec{L}+\vec{Z}$, where
$\vec{Z}\in\R^m$ and $Z_i=\max_{j:\sg(j)=i}p_{ij}$; the filtering step and our
constraints ensure that $f(\vec{Z})\leq 2T$, so $f(2\vec{L}+\vec{Z})\leq 4T$.
Our algorithm is much more direct than the one in~\cite{ChakrabartyS19}: it avoids
the sequence of steps (and the associated approximation-factor losses) used
in~\cite{ChakrabartyS19}, wherein \minlb is reduced to a special case, called \minmax
ordered load balancing, which is then tackled by a deterministic oblivious rounding
procedure.

\paragraph{Other related work.} 
The algorithmic problem of finding minimum-norm solutions 
has also been investigated in the context of {\em $k$-clustering}, wherein the goal is to
open $k$ ``facilities'' in a metric space to serve a set of clients, and the cost vector 
induced by a solution
is the vector of distances of clients to their nearest open facility.
The setting of $\ell_p$-norms, especially when $p\in\{1,2,\infty\}$ (where the
problem is called the $k$-\{median,\,means,\,center\} problem) has been
extensively studied, and $O(1)$-approximations are known in these
settings~\cite{HochbaumS85a,CharikarGST02,JainV01,AhmadianNSW17}. 
\topl and ordered norms have been proposed in the context of $k$-clustering in the
Operations Research literature (see, e.g.,~\cite{NickelP05,LaporteNdG15}), but
constant-factor approximations for these norms 
were obtained quite recently~\cite{ByrkaSS18,ChakrabartyS18,ChakrabartyS19}.
Furthermore, Chakrabarty and Swamy~\cite{ChakrabartyS19} utilize their general framework
to obtain an $O(1)$-approximation for the $k$-clustering problem under any monotone,
symmetric norm. 
We do not know of any alternate approach that works in the $k$-clustering setting.

\section{A convex-programming relaxation} \label{convprog}
By possibly adding dummy jobs with zero processing times, we may assume without loss of
generality that $n\geq m$. 
A natural convex program for \minlb has non-negative variables $x_{ij}$ denoting if job $j$ is
assigned to machine $i$ (or the extent of $j$ assigned to $i$) with the constraint
\eqref{jasgn} encoding that every job is assigned to a machine.
These $x$-variables 
define a load vector $\vec{L}=\bigl(L_i=L_i(x)\bigr)_{i\in[m]}$ where $L_i(x) = \sum_{j\in J} p_{ij}x_{ij}$.
The objective seeks to minimize
$T:=f\bigl(\vec{L}\bigr)$.
As noted earlier, this convex program has a large integrality gap (even when
$f$ is the $\ell_\infty$ norm). 
We {\em strengthen} the convex program as follows. 

Given the $x$-assignment, define $P_j = P_j(x):=\sum_i p_{ij}x_{ij}$, which is the load
incurred by the fractional solution for scheduling job $j$. 
Fix any subset $S\sse J$ with $|S| = m$. Note that this is well-defined since we have assumed $n\geq m$.
This defines the $m$-dimensional vector $\vec{P}_S:=\{P_j\}_{j\in S}$. We add the
constraints \eqref{eq:jobs} enforcing that
$f(\vec{P}_S) \leq T$ for each such subset $S$. 
Throughout, we use $i$ to index the machines in $[m]$, and $j$ to index the jobs in $J$.
\begin{alignat}{3}
\min & \quad & T & \tag{CP} \label{eq:cp} \\
\text{s.t.} & \quad & \sum_i x_{ij}  & \ge 1 \qquad && \forall j\in J \label{jasgn} \\[-0.75ex]
&& x & \geq 0 \label{nonneg} \\
&& L_i & = \sum_{j\in J} p_{ij}x_{ij} \qquad && \forall i\in [m] \label{mcload} \\
&& P_j & = \sum_{i\in[m]} p_{ij}x_{ij} \qquad && \forall j\in J \label{jbload} \\
&& f(\vec{L}) & \leq T \label{eq:machines}\\
&& f(\vec{P}_S) & \leq T \qquad && \forall S \subseteq J: |S| = m \label{eq:jobs}
\end{alignat}
Let $\OPT:=\OPT_{\text{\ref{eq:cp}}}$ denote the optimal value of \eqref{eq:cp}, and let
$\iopt$ be the optimal value of the minimum-norm load-balancing problem. 
Since the $x_{ij}$-variables completely determine a solution to \eqref{eq:cp}, we will
sometimes abuse notation and say that $x$ is a feasible solution to \eqref{eq:cp}.
We argue that \eqref{eq:cp} is a valid relaxation. The proof uses the following simple
observation about symmetric convex functions.

\begin{claim}\label{clm:simple}
Let $h:\R^m \to \R$ be a symmetric convex function.
Let $v\in \R^m_+$, and $i,j\in [m]$. Let $w\in\R^m_+$ be the vector where 
$w_i = v_i+v_j$, $w_j = 0$, and $w_k = v_k$ otherwise. Then, $h(v) \leq h(w)$. 
\end{claim}

\begin{proof}
Consider the vector $w'$ constructed in a symmetric fashion to $w$: set 
$w'_j = v_i + v_j$, $w'_i = 0$, and $w'_k = v_k$ otherwise. 
Observe that $v$ is a convex combination of $w$ and $w'$ (we have 
$v=\frac{v_i}{v_i+v_j}\cdot w+\frac{v_j}{v_i+v_j}\cdot w'$), and $h(w)=h(w')$ since $h$ 
is symmetric.
By convexity and symmetry, $h(v)\leq\max\bigl\{h(w), h(w')\bigr\}=h(w)$. 
\end{proof}

\begin{theorem}\label{thm:valid}
Constraints \eqref{eq:jobs} are valid, and so for any instance of \minlb, we have
$\OPT\leq\iopt$. 
\end{theorem}

\begin{proof}
Let $\sg^*:J\to [m]$ be an optimal assignment, so $f(\lvec_{\sg^*}) = \iopt$. We now
describe a feasible solution to \eqref{eq:cp} with $T=\iopt$. 
Set $x_{ij} = 1$ if $\sg^*(j) = i$, and $0$ otherwise. Clearly, constraints \eqref{jasgn}
hold. Note, $L_i = \load_{\sg^*}(i)$ for all $i$, and $P_j = p_{\sg^*(j)j}$ for all $j$. 
Therefore, \eqref{eq:machines} holds with equality. 
	
The interesting bit is to show that~\eqref{eq:jobs} holds. To that end, fix a subset
$S\sse J$ of $m$ jobs. Consider the load vector induced by jobs in $S$. That is, define
$L'_i := \sum_{j\in S:\sg^*(j) = i} p_{ij}$.
Note that $\vec{L}$ coordinate wise dominates $\vec{L'}$, so 
by monotonicity of $f$, we have $f(\vec{L'})\leq f(\vec{L})=T$.  

We argue that $f(\vec{P}_S)\leq f(\vec{L'})$, which will complete the proof. 
To see this, first note that if $\sg^*$
assigns the jobs in $S$ to distinct machines, then $\vec{P}_S$ is simply a 
permutation of $\vec{L'}$, so $f(\vec{P}_S)=f(\vec{L'})$. Otherwise, observe that
$\vec{L'}$ can be obtained from $\vec{P}_S$ by applying the operation in
Claim~\ref{clm:simple} to pairs of jobs in $S$ assigned to the same machine; 
therefore, we have $f(\vec{P}_S)\leq f(\vec{L'})$. 
\end{proof}

The proof above relied only on convexity, monotonicity, and symmetry of the function
$f$. 
In Section~\ref{rounding} (see Theorem~\ref{thm:rnd}) we describe a rounding procedure
which takes a feasible solution for \eqref{eq:cp} and returns an assignment with a
factor-$4$ blow-up in the objective. This will utilize the homogeneity of the norm $f$.  
In Section~\ref{convsolve}, we show how to (approximately) solve \eqref{eq:cp} given
an (approximate) first-order oracle for the underlying norm
(see Theorem~\ref{thm:cpsolve}). Combining these two results yields Theorem~\ref{4apx}.

\section{The rounding algorithm} \label{rounding}
We now describe and analyze our simple rounding algorithm, which yields the
following guarantee. 

\begin{theorem} \label{thm:rnd}
Given a feasible fractional solution
$\bigl(x=\{x_{ij}\}_{i,j},\vec{L},\vec{P},T\bigr)$ to
\eqref{eq:cp}, there is a polynomial time algorithm to obtain a schedule $\sigma$ with 
$f(\lvec_\sg)\leq 4T$.
\end{theorem}

\begin{proof}
First, we filter $x$. For every $i,j$, we set $\hx_{ij} = 2x_{ij}$ if $p_{ij}\leq 2P_j$, 
and $0$ otherwise. A standard Markov-inequality style argument shows that  $\hx$ satisfies \eqref{jasgn}.
Now we apply the Shmoys-Tardos GAP-rounding algorithm~\cite{ShmoysT93} to $\hx$. This yields an
assignment $\sg:J\to[m]$ such that: for every job $j$, we have $\sg(j)=i$ only if
$\hx_{ij}>0$, and
for every machine $i$, we have
$\load_\sg(i)\leq \sum_{j\in J} p_{ij}\hx_{ij} + Z_i \leq 2L_i + Z_i$,
where $Z_i = \max_{j: \sg(j) = i} p_{ij}$. Thus, 
$\lvec_{\sg} \leq 2\vec{L} + \vec{Z}$. 

Let $j_i$ be a maximum-length job assigned to
machine $i$ in $\sg$, i.e., $\sg(j_i)=i$ and $Z_i=p_{ij_i}$. By our filtering step, we
know that $Z_i \leq 2P_{j_i}$. 
Let $S = \{j_i: i\in[m]\}$. Then $\vec{Z}:=(Z_i)_{i\in[m]}\leq 2\vec{P}_S$. 
By monotonicity, the triangle inequality, and
homogeneity of $f$, 
we then obtain that
\begin{equation*}\label{eq:algo}
f(\lvec_{\sg}) \leq 2f(\vec{L}) + f(\vec{Z}) \leq 2T + 2f(\vec{P}_S)\leq 4T. \qedhere
\end{equation*}
\end{proof}

Interestingly, and notably, observe that the rounding procedure above is 
{\em oblivious to the norm $f$}: given a fractional solution $x$, the same rounding
procedure works for all monotone, symmetric norms. This will be useful in
Section~\ref{multinorm}, where we seek an assignment that is simultaneously good for
multiple norms.

\section{Solving the convex program} \label{convsolve}
We now discuss how to solve the convex program~\eqref{eq:cp}. It is well
known~\cite{NemirovskiY76,GrotschelLS88} that we can
efficiently solve a convex program $\min_{x\in S}h(x)$ 
(where $S\sse\R^n$ is convex) 
to within any additive error $\ve>0$ using the
ellipsoid method provided that (we state things more precisely below): (i) $S$ has
non-zero volume and is contained in some ball; (ii) we have a separation oracle for $S$;
(iii) we have a {\em first-order} oracle for $h$ that given input $x\in S$, returns
$h(x)$, and a subgradient of $h$ at $x$. 
More generally, we show that by utilizing the machinery of Shmoys and
Swamy~\cite{ShmoysS06}, even an approximate value and subgradient oracle suffices (see
Theorem~\ref{thm:cpsolve}). 
This is particularly relevant since the norm and/or components of the subgradient vector 
may involve irrational numbers. 

By scaling we may assume that all $p_{ij}$s are integers. Let $\iopt$ denote the optimal
value for the \minlb instance. We can easily detect if $\iopt=0$, since this implies an
assignment with $0$ load on every machine.
Therefore, we assume $\iopt \ge 1$.
It will be convenient to reformulate \eqref{eq:cp} as follows.
Let $\Pc:=\bigl\{x\in\R^{[m]\times J}:\ \sum_i x_{ij} \ge 1\ \ \forall j\in J,\ \ 
0\leq x_{ij}\leq 1\ \ \forall i\in[m],j\in J\bigr\}$
denote the feasible region for the assignment variables.
\begin{equation}
\min \ \ g(x)\ :=\ \max\Bigl\{f\bigl(\vec{L(x)}\bigr),\ \ \max_{S\sse J:|S|=m}f\bigl(\vec{P(x)}_S\bigr)\Bigr\}
\qquad \text{s.t.} \qquad x\in \Pc.
\tag{CP'} \label{newcp}
\end{equation}
Note that the $x_{ij}$s are the only variables above.
Recall that $\OPT$ is the optimal value of \eqref{eq:cp} (and \eqref{newcp}).

We recall a few standard concepts from optimization.
Let $h:\R^k\mapsto\R$ and 
let $\|u\|$ denote the $\ell_2$ norm of $u$.
\begin{enumerate}[label=\textbullet, topsep=0.5ex, itemsep=0ex,
    labelwidth=\widthof{\textbullet}, leftmargin=!]
\item We say that $h$ has {\em Lipschitz constant} (at most) $K$ if 
$|h(v)-h(u)|\leq K\|v-u\|$ for all $u,v\in\R^{k}$.

\item 
We say that $\sgr\in\R^{k}$ is a {\em subgradient} of $h$ at
$u\in\R^{k}$ if we have $h(v)-h(u)\geq\sgr\cdot(v-u)$ for all 
$v\in\R^{k}$. 
We say that $\hsgr$ is an {\em $\w$-subgradient} of 
$h$ 
at $u\in\R^k$ if for every $v\in\R^k$, we have $h(v)-h(u)\geq\hsgr\cdot(v-u)-\w h(u)$; we
call this the approximate-subgradient inequality.

\item 
An {\em $\w$-first-order oracle} for $h$ is an
algorithm that at any point $u\in\R^k$, returns an estimate $\est$ such that
$h(u)\leq\est\leq(1+\w)h(u)$, and an $\w$-subgradient of $h$ at $u$.

(In the optimization literature, the notions of approximate first-order oracle
and approximate subgradient typically involve additive errors; since our problems are 
scale-invariant, multiplicative approximations, where the error at $u$ is measured
relative to $h(u)$, are more apt here.)
\end{enumerate}
\noindent
We remark that since $f$ is a norm, an $\w$-subgradient $\hsgr$ of $f$ at $u$ also
yields an estimate of $f(u)$ as follows: 
taking $v=\vec{0}$ and $v=2u$ respectively in the approximate-subgradient inequality, we
obtain the bounds $\hsgr\cdot u\geq(1-\w)f(u)$ and $\hsgr\cdot u\leq(1+\w)f(u)$.
(Thus, an $\w$-first-order oracle for $f$ boils down to an $\w$-subgradient oracle for
$f$.)

By $\inpsize$, we mean the total encoding length of the $p_{ij}$s.
It is easy to separate over $\Pc$, and easy to find radii $R$, and $0<V\leq 1$ such
that $\Pc\sse B(\bo,R):=\{x: \|x\|\leq R\}$,
$\Pc$ contains a ball of radius $V$, and $\log\bigl(\frac{R}{V}\bigr)=\poly(m,n)$.
In particular, $R = \sqrt{mn}$ suffices, and $\Pc$ contains a ball of radius 
$V=\frac{0.5}{m}$ around the point $x$ with $x_{ij}=\frac{1.5}{m}$ for all $i,j$.
(We may assume $m\geq 2$ as otherwise the problem is trivial.)
Throughout, we use $K_f$ to denote an efficiently-computable upper bound on the Lipschitz
constant of $f$; Lemma~\ref{gprops} shows how to obtain this.
Given a bound on the Lipschitz constant of $f$, one can compute an upper bound on the
Lipschitz constant of $g$. 

\begin{claim} \label{glip}
The Lipschitz constant of $g$ is at most $K=\sqrt{mn}\cdot\max_{i,j}p_{ij}\cdot K_f$. 
\end{claim}

\begin{theorem}[Follows from~\cite{NemirovskiY76}; see also~\cite{GrotschelLS88}] 
\label{cvoptexact}
Let $\alg$ be a first-order oracle for $f$. 
Then, for any $\eta>0$, we can compute 
$x^*\in\Pc$ such that $g(x^*)\leq\OPT+\eta$ in 
$\poly\bigl(\inpsize,\log(\frac{K_fR}{\eta V})\bigr)$ time and using 
$\poly\bigl(\inpsize,\log(\frac{K_fR}{\eta V})\bigr)$ calls to 
$\alg$.
\end{theorem}

Theorem~\ref{cvoptexact} follows from the ellipsoid method for convex optimization, 
due to the bound on the Lipschitz constant of $g$ obtained from Claim~\ref{glip}, and
since one can use $\alg$ to obtain a first-order oracle for $g$.  
We next use~\cite{ShmoysS06} to obtain a stronger result that utilizes only an
approximate first-order oracle for $f$. 

\begin{theorem}[Lemma 4.5 in~\cite{ShmoysS06} paraphrased] \label{sscvopt}
Consider a convex optimization problem: $\min_{x\in\Pc} h(x)$.
Let $K_h$ be a known bound on the Lipschitz constant of $h$. 
Let $\w<1$ and $\eta>0$. 
In $\poly\bigl(m,n,\log(\frac{K_hR}{V\eta})\bigr)$ time
and using $\poly\bigl(m,n,\log(\frac{K_hR}{V\eta})\bigr)$ calls to an $\w$-first-order
oracle for $h$, one can compute a solution $x^*\in\Pc$ such that
$h(x^*)\leq\frac{1+\w}{1-\w}\cdot\bigl(\min_{x\in\Pc}h(x)+\eta\bigr)$.  
\end{theorem}

To utilize Theorem~\ref{sscvopt} to solve \eqref{eq:cp}, we show how to obtain an
approximate first-order oracle for $g$ given one for $f$. Also, in order to convert the
additive error in Theorem~\ref{sscvopt} (and Theorem~\ref{cvoptexact}) into a multiplicative
guarantee,  
we show how to obtain a lower bound $\lb$ on $\iopt$ such that $K_f/\lb$ is small.

\begin{lemma} \label{gprops}
Let $\alg$ be an $\w$-first-order oracle for $f$ (where $\w<1$). 
\begin{enumerate}[(i), topsep=0ex, itemsep=0ex, labelwidth=\widthof{(iii)}, leftmargin=!]
\item \label{gfo}
We can obtain a $2\w$-first-order oracle for $g$ using $O(1)$ calls to $\alg$. 
\item \label{lbound} Using $\alg$, 
we can efficiently compute $\lb\leq\iopt$, and an upper bound $K_f$ on the
Lipschitz constant of $f$ such that 
$\frac{K_f}{\lb}\leq 2\sqrt{m}$. 
\end{enumerate}
\end{lemma}

\begin{theorem} \label{thm:cpsolve}
Let $\alg$ be an $\w$-first-order oracle for $f$ with $\w\leq\frac{1}{10}$.
Given a \minlb instance with optimum value $\iopt$, there is an algorithm that,
for any $\ve>0$, computes a feasible solution $x^*$ to
\eqref{eq:cp} of objective value $g(x^*)\leq(1+5\w)(1+\ve)\iopt$.
The algorithm runs in $\poly\bigl(\inpsize,\log(\frac{1}{\ve})\bigr)$ time 
and makes $\poly\bigl(\inpsize,\log(\frac{1}{\ve})\bigr)$ 
calls to $\alg$.
\end{theorem}

\begin{proof}
This follows by combining Theorem~\ref{sscvopt} and Lemma~\ref{gprops}. 
Recall that we are assuming that $\iopt\geq 1$.
By part \ref{gfo} of Lemma~\ref{gprops}, we can compute a $2\w$-first-order oracle for $g$.
We use part \ref{lbound} of Lemma~\ref{gprops} to obtain $\lb$ and $K_f$. Now we apply
Theorem~\ref{sscvopt} to the problem $\min_{x\in\Pc}g(x)$, taking $\eta=\ve\lb$. 
The point $x^*$ returned satisfies 
$g(x^*)\leq\frac{1+2\w}{1-2\w}\cdot(\OPT+\ve\lb)\leq(1+5\w)(1+\ve)\iopt$. 

Recall that $\log(R/V)=\poly(m,n)$. Since we have an upper bound $K$ on the
Lipschitz constant of $g$, where 
$\log K=\poly(\inpsize)\cdot\log K_f$ (Claim~\ref{glip}), the running time and 
number of calls to the first-order oracle for $g$ (and hence $\alg$) is at most
$\poly\bigl(\inpsize,\log(\frac{1}{\ve})\bigr)$. 
\end{proof}

\section{Extensions: multi-norm load balancing and simultaneous approximation} 
\label{multinorm} \label{extn}

\subsection{Multi-norm load balancing}
In the {\em multi-norm load-balancing problem}, 
we are given a load-balancing instance $\bigl(J,m,\{p_{ij}\}_{i\in[m],j\in J}\bigr)$,
multiple monotone, symmetric norms 
$f_1,\ldots,f_k$, and budgets $T_1,\ldots,T_k$ for
these norms respectively. The goal is to find an assignment $\sg:J\to[m]$ such that 
$f_r(\lvec_\sg)\leq T_r$ for all $r\in[k]$.
Our approximation guarantee extends easily to this problem. 

\begin{theorem} \label{mnorm}
Let $\bigl(J,m,\{p_{ij}\}_{i\in[m],j\in J}\bigr)$ be a load-balancing instance. 
Let $f_1,\ldots, f_k$ be $k$ monotone, symmetric norms, with associated budgets $T_1,\ldots, T_k$.
Given an $\w$-first-order oracle for each norm, for any $\ve > 0$, in
$\poly(\inpsize,k,\log(1/\ve))$ time, one can either determine that there is no feasible
solution to the multi-norm load-balancing problem, or return an assignment $\sg:J\to [m]$
such that $f_r(\lvec_\sg) \leq 4(1+7\w)(1+\ve)T_r$ for all $r\in[k]$. 
\end{theorem}

The convex-programming relaxation for this problem is a variant of \eqref{eq:cp} 
where there is no objective function, and
constraints \eqref{eq:machines}, \eqref{eq:jobs} are replaced with
\begin{equation}
f_r(\vL)\leq T_r, \qquad 
f_r(\vP_S)\leq T_r \quad \forall S\sse J: |S|=m, \qquad \quad \forall r=1,\ldots,k
\label{budget}
\end{equation}
\newcommand{\multicp}{(Multi-CP)\xspace}
Let \multicp denote the resulting feasibility problem: find $(x,\vL,\vP)$
satisfying \eqref{jasgn}--\eqref{jbload}, and \eqref{budget}. 
As noted earlier, the rounding procedure in Section~\ref{rounding} is {\em oblivious} to
the underlying norm, and so our task boils down to finding an (approximately) feasible
solution to \multicp.

In order to 
solve \multicp, as with \eqref{eq:cp}, it will be
convenient to move the nonlinear constraints to the objective and consider the following
reformulation: 
\begin{equation}
\min \ \ \mnp(x)\ :=\ \max\Bigl\{\max_{r\in[k]}\tfrac{f_r(\overrightarrow{L(x)})}{T_r},
\ \ \max_{r\in[k]}\max_{S\sse J:|S|=m}\tfrac{f(\overrightarrow{P(x)}_S)}{T_r}\Bigr\}
\qquad \text{s.t.} \qquad \eqref{jasgn}, \eqref{nonneg}. \tag{MNCP} \label{mncp}
\end{equation}
Observe that finding a feasible solution to \multicp is equivalent to finding a feasible
solution to \eqref{mncp} with objective value at at most $1$.
As before, we may assume that the $p_{ij}$s are integers, and can determine if there is an 
assignment $\sg$ such that $\lvec_\sg=\vec{0}$ (which clearly satisfies \eqref{budget}). 
So assume otherwise. We prove the following. 

\begin{theorem} \label{mncpsolve}
Let $\alg_r$ be an $\w$-first-order oracle for $f_r$ for all $r\in[k]$, where
$\w\leq\frac{1}{18}$. 
For any $\ve>0$, in $\poly\bigl(\inpsize,\log(\frac{1}{\ve})\bigr)$ time and using
$\poly\bigl(\inpsize,\log(\frac{1}{\ve})\bigr)$ calls to each $\alg_r$ oracle, we can
determine that either \multicp is infeasible, or compute $x^*\in\Pc$ such that
$\mnp(x^*)\leq(1+7\w)(1+\ve)$. 
\end{theorem}

Using Theorem~\ref{mncpsolve}, 
for any $\ve>0$, 
we can determine in time $\poly\bigl(\inpsize,\log(\frac{1}{\ve}\bigr)$ that
\multicp is infeasible, or return a 
fractional assignment $x^*$ satisfying
\[
f_r(\vL[x^*])\leq \kp T_r, \qquad 
f_r(\vP[x^*]_S)\leq \kp T_r \quad \forall S\sse J: |S|=m, \qquad \quad \forall r=1,\ldots,k
\]
where $\kp=(1+7\w)(1+\ve)$.
As noted earlier, the rounding procedure in Section~\ref{rounding} is {\em oblivious} to
the underlying norm, and so by utilizing this to round $x^*$, we obtain an assignment
$\sg$ such that $f_r(\lvec_\sg)\leq 4\kp T_r$ for all $r\in[k]$. This yields
Theorem~\ref{mnorm}. 

\medskip
In the rest of this section, we discuss the proof of Theorem~\ref{mncpsolve}.
If the multi-norm problem is feasible, we must have $T_r\geq f_r(e_1)$ for all $r\in[k]$.  
We assume in the sequel that $T_r$ is at least the estimate of $f_r(e_1)$ returned by
$\alg_r$ scaled by $(1+\w)$, for all $r\in[k]$; if this does not hold, then we declare
infeasibility. Given this, the proof of Lemma~\ref{gprops} \ref{lbound} shows that 
$K_r=(1+\w)\sqrt{m}\cdot T_r$ is an upper bound on the Lipschitz constant of $f_r$, for
all $r\in[k]$. We assume this bound in the sequel. 
Similar to Claim~\ref{glip} and Lemma~\ref{gprops}, we show that the Lipschitz constant of
$\mnp$ can be bounded in terms of the $K_r$s, and we can obtain a $2\w$-first-order oracle
for $\mnp$ using the $\alg_r$ oracles. 

\begin{lemma} \label{mnprops}
(i) The Lipschitz constant of $\mnp$ is bounded by $K=\poly(m,n,\max_{i,j}p_{ij})$.
(ii) We can obtain a $2\w$-first order oracle for $\mnp$ by making $O(1)$ calls to $\alg_r$ for
each $r\in[k]$. 
\end{lemma}

\begin{proofof}{Theorem~\ref{mncpsolve}}
	We utilize Lemma~\ref{mnprops} in conjunction with Theorem~\ref{sscvopt}. 
	Part (ii) of Lemma~\ref{mnprops} shows how to obtain a $2\w$-first-order oracle, $\alg$,
	for $\mnp$. So invoking Theorem~\ref{sscvopt} with $\eta=\ve$, and the bound $K$ on the
	Lipschitz constant of $\mnp$ obtained from part (i) of Lemma~\ref{mnprops}, we obtain 
	$\bx\in\Pc$ such that 
	\begin{equation}
	\mnp(\bx)\leq\frac{1+2\w}{1-2\w}\Bigl(\min_{x\in\Pc}\mnp(x)+\eta\Bigr).
	\label{bnd1}
	\end{equation}
	The running time is $\poly\bigl(\inpsize,\log(\frac{1}{\ve})\bigr)$ (since 
	$\log(R/V)$, $\log K=\poly(\inpsize)$), and this is also a bound on the number of calls   
	to the $\alg_r$ oracles. 
	Using $\alg$, we obtain an estimate $\est$ such that
	$\mnp(\bx)\leq\est\leq(1+2\w)\mnp(\bx)$. If $\est>\frac{(1+2\w)^2}{1-2\w}\cdot(1+\eta)$,
	then \eqref{bnd1} implies that $\bigl(\min_{x\in\Pc}\mnp(x)\bigr)>1$, and so \multicp is
	infeasible. 
	Otherwise, taking $x^*=\bx$, we obtain that 
	$\mnp(x^*)\leq\est\leq\frac{(1+2\w)^2}{1-2\w}\cdot(1+\ve)\leq(1+7\w)(1+\ve)$ since
	$\w\leq\frac{1}{18}$.
\end{proofof}

\subsection{Simultaneous approximation}\label{sec:simul}
Given a load-balancing instance $\I=\bigl(J,m,\{p_{ij}\}_{i\in[m],j\in J}\bigr)$, let
$\al^*_\I$ be the smallest $\al$ such that there {\em exists} an assignment $\sg^*$
satisfying $f(\lvec_{\sg^*})\leq\al\bigl(\min_{\sg:J\mapsto[m]}f(\lvec_\sg)\bigr)$ for
{\em every} monotone, symmetric norm.
That is, $\al^*_\I$ is the best {\em simultaneous approximation factor} achievable on
instance $\I$. 
Instead of seeking absolute bounds on $\al^*_\I$ over a class of 
instances~\cite{AlonAWY98,AzarERW04,GoelM06}, as discussed 
in~\cite{ChakrabartyS19}, another pertinent problem is to seek {\em instance-wise}
guarantees:  
given an instance $\I$, we want to find a polytime-computable assignment $\bsg$ 
such that, for some factor $\gm\geq 1$, we have
$f(\lvec_{\bsg})\leq\gm\al^*_\I\bigl(\min_{\sg:J\mapsto[m]}f(\lvec_\sg)\bigr)$ for every
monotone, symmetric norm; i.e., the simultaneous approximation factor of $\bsg$
at most $\gm$ times the best simultaneous approximation factor achievable for
$\I$. 

Our techniques coupled with insights from~\cite{GoelM06,ChakrabartyS19} yields
{\em a $4\bigl(1+O(\ve)\bigr)$-approximation to the best simultaneous approximation factor}, 
in time $\poly\bigl(\inpsize,(\frac{m}{\ve})^{O(1/\ve)}\bigr)$. 
To obtain this guarantee,
following~\cite{GoelM06,ChakrabartyS19}, incurring a $(1+\ve)$-factor loss, it suffices to
obtain a $4$-approximation to the 
best simultaneous-approximation achievable for \topl-norms---%
$\topl(x):=\max_{S\sse[m]:|S|=\ell}\sum_{i\in S}|x_i|$---%
for the $O(\log m)$ indices $\ell$ in 
$\pset:=\bigl\{\min\{\ceil{(1+\ve)^s},m\}: s\geq 0\bigr\}$. If we knew the optimal value
$\optl$ for each such \topl norm, then we can set set a budget $T_\ell=\al\optl$ for each
$\ell\in\pset$, and utilize our result for multi-norm load balancing to do a binary 
search for $\al$. Importantly, notice that the resulting feasibility problem \multicp 
can now be cast as an {\em linear-programming} feasibility problem, since a budget
constraint of the form $\topl(\vec{v})\leq T_\ell$ can be modeled using exponentially
many linear constraints that one can separate over. Thus, this would yield a
$4(1+\ve)$-approximation. 
To make this idea work, we enumerate all choices for the $\optl$ values in powers of
$(1+\ve)$. As argued in~\cite{ChakrabartyS19}, there are at most
$\poly\bigl(\inpsize,(\frac{m}{\ve})^{O(1/\ve)}\bigr)$ candidates to enumerate over, and
this yields the stated guarantee. 

\begin{theorem}\label{thm:simul-opt}
Given a load-balancing instance $\I = \bigl(J,m,\{p_{ij}\}_{i\in[m],j\in J}\bigr)$, let $\alpha^*_I$
be the smallest $\alpha$ such that there is an assignment $\sg^*$ satisfying
$f(\lvec_{\sg^*})\leq\al\bigl(\min_{\sg:J\mapsto[m]}f(\lvec_\sg)\bigr)$ for every
monotone, symmetric norm $f$. 
In
$\poly\bigl(\inpsize,(\frac{m}{\ve})^{O(1/\ve)}\bigr)$ time, we can find an assignment 
$\widehat{\sg}$ such that we have \linebreak $f(\lvec_{\widehat{\sg}})\leq
\bigl(4+O(\ve)\bigr)\al^*_\I\bigl(\min_{\sg:J\mapsto[m]}f(\lvec_\sg)\bigr)$ for every
monotone, symmetric norm $f$.
\end{theorem}

\bibliographystyle{plain}
\bibliography{minnorm-lb}

\appendix

\section{Proofs from Sections~\ref{convsolve} and~\ref{extn}}

\begin{proofof}{Claim~\ref{glip}}
The bound follows easily from the definition of $g$. Let $x,y\in\R^{[m]\times J}$. Let
$\vL, \vec{L'}\in\R^m$ be the load vectors induced by $x, y$ respectively; let $\vP_S$,
$\vec{P'}_S$, be the job-cost vectors for the jobs in $S$ induced by $x, y$ respectively.  
Then, $g(y)-g(x)\leq\max\bigl\{f(L')-f(L),\max_{S\sse J:|S|=m}f(P'_S)-f(P_S)\bigr\}$.
So $g(y)-g(x)\leq K_f\|L'-L\|_2$ or $g(y)-g(x)\leq K_f\|P'_S-P_S\|$ for some $S\sse J$
with $|S|=m$. Let $p_{\max}:=\max_{i,j}p_{ij}$.
In the former case, we have 
$g(y)-g(x)\leq K_fp_{\max}\sum_{i,j}|y_{ij}-x_{ij}|\leq\sqrt{mn}\cdot K_fp_{\max}\|y-x\|_2$; 
the same bound also applies in the latter case. This shows shows that
$K=\sqrt{mn}\cdot K_fp_{\max}$ is a bound on the Lipschitz constant of $g$.
\end{proofof}

The following claim will be useful in proving part~\ref{gfo} of Lemma~\ref{gprops}, as
also part (ii) of Lemma~\ref{mnprops}.

\begin{claim} \label{fotrans}
Let $h:\R^N\mapsto\R$ be defined by $h(x):=\max_{r\in[k]}h_r(x)$, where
$h_r:\R^N\mapsto\R$ is convex for all $r\in[k]$. 
Let $\alg_r$ be an $\w$-first order oracle for $h_r$ for all $r\in[k]$ (where $\w<1$). 
\begin{enumerate}[(i), topsep=0.5ex, itemsep=0ex, labelwidth=\widthof{(ii)}, leftmargin=!]
\item One can obtain a $2\w$-first order oracle for $h$ using $O(1)$ calls to
$\alg_1,\ldots\alg_k$. 

\item More generally, suppose that given $x\in\R^n$, one can identify $I(x)\sse[k]$ such
that $h(x)=\max_{r\in I(x)}h_r(x)$. Then, one can compute a $2\w$-first-order oracle for $h$
that, on input $x\in\R^n$, makes $O(1)$ calls to $\alg_r$ for all $r\in I(x)$.
\end{enumerate}
\end{claim}

\begin{proof}
We focus on proving part (i); part (ii) follows from a very similar argument.
Fix $x\in\R^N$. For every $r\in[k]$, we call $\alg_r$ to obtain an estimate $\est^r$ of
$h_r(x)$. We set the estimate for $h(x)$ to be $\est:=\max_{r\in[k]}\est^r$. From the
properties of $\est^r$, it is easy to see that $h(x)\leq\est\leq(1+\w)h(x)$.

Let $\sgr^r$ be the $\w$-subgradient of $f_r$ at $x$ returned by $\alg_r$.
Let $s\in[k]$ be such that $\est=\est^s$. We set $\mu=\sgr^s$. We now argue that $\mu$ is
a $2\w$-subgradient of $h$ at $x$. Consider any $y\in\R^N$. 
We have
\begin{equation*}
\begin{split}
\mu^T(y-x)&=(y-x)^T\sgr^s\leq h_s(y)-h_s(x)+\w h_s(x)
\leq h(y)-\tfrac{1-\w}{1+\w}\cdot\est^s
= h(y)-\tfrac{1-\w}{1+\w}\cdot\est \\
& \leq h(y)-\tfrac{1-\w}{1+\w}\cdot h(x)\leq h(y)-(1-2\w)h(x).
\end{split}
\end{equation*}
The first two inequalities follow due to the fact that $(\est^s,\sgr^s)$ was returned by
the $\w$-first order oracle for $h_s$;
the next equality follows from the definition of index $s$; and the penultimate inequality 
follows since $\est\geq h(x)$ as established earlier.

\medskip
The proof of the more general statement in (ii) is essentially identical: on input $x$, we
now run $\alg^r$ for all $r\in I(x)$; we set $\est=\max_{r\in I(x)}\est^r$, and $\sgr=\sgr^s$,
where $s\in I(x)$ is an index such that $\est=\est^s$.
\end{proof}

\begin{proofof}{Lemma~\ref{gprops}}
For part~\ref{gfo}, fix $x\in\R^{[m]\times J}$. Recall that 
$P_j=P_j(x):=\sum_i p_{ij}x_{ij}$,
Let $S^*$ be the set of $m$ jobs with the highest $P_j$ values. 
Let $\vL=\vL[x]$ and $\vec{P}_{S^*}=\vec{P(x)}_{S^*}$.
Then, $g(x)=\max\bigl\{f(\vec{L}),f(\vec{P}_{S^*})\bigr\}$.
Observe that $\alg$ can be used to obtain an $\w$-first-order oracle for both
$f\bigl(\vL[x]\bigr)$ and $f\bigl(\vP[x]_{S^*}\bigr)$.
Thus, by using Claim~\ref{fotrans} (ii), we obtain a $2\w$-first-order oracle for $g$
using $O(1)$ calls to $\alg$.

We now justify the observation.
A $(1+\w)$-approximate value oracle is obtained by simply calling $\alg$ to
obtain estimates of $f(\vL)$ and $f(\vP_{S^*})$.
Let $\sgr^L=(\sgr^L_i)_{i\in[m]}$, and $\sgr^P=(\sgr^P_j)_{j\in S^*}$ be the
$\w$-subgradients of $f$ at $\vec{L}$ at $\vec{P}_{S^*}$ respectively returned by $\alg$. 
\[
\text{For all $i\in[m], j\in J$, define} \qquad 
\beta_{ij}=p_{ij}\sgr^L_i, \qquad 
\gm_{ij}=
\begin{cases}p_{ij}\sgr^P_j & \text{if $j\in S^*$}; \\ 0 & \text{otherwise.}\end{cases}
\]
Then, for any $y\in\R^{[m]\times J}$, we have
$\beta^T(y-x)=\sum_{i,j}\sgr^L_ip_{ij}(y_{ij}-x_{ij})=(\vL[y]-\vL[x])^T\sgr^L$
showing that $\beta$ is an $\w$-subgradient of $f\bigl(\vL[\cdot]\bigr)$ at $x$.
Similarly, $\gm^T(y-x)=\bigl(\vP[y]_{S^*}-\vP[x]_{S^*}\bigr)^T\sgr^P$
showing that $\gm$ is an $\w$-subgradient of $f\bigl(\vP[\cdot]\bigr)$ at $x$.

\medskip
For part~\ref{lbound},
Let $\sg^*$ be an optimal assignment. Since we are assuming that $\iopt\geq 1$, 
we have $\load_{\sg^*}(i)\geq 1$ for some $i\in[m]$. Let $e_i\in\R^m$ be the vector with
$1$ in coordinate $i$ and $0$s everywhere else. Then, $\iopt\geq f(e_1)$. Let $\lb$ be
the estimate of $f(e_1)$ obtained by $\alg$ scaled down by $(1+\w)$. So we have
$f(e_1)/(1+\w)\leq\lb\leq\iopt$. 
Consider any $x,y\in\R^m$. We have $y=x+\sum_{i=1}^m(y_i-x_i)e_i$, so by the triangle
inequality and symmetry, we have $|f(y)-f(x)|\leq\sum_{i=1}^m|y_i-x_i|f(e_i)$
Therefore, $|f(y)-f(x)|\leq(1+\w)\lb\sum_{i=1}^m|y_i-x_i|\leq(1+\w)\sqrt{m}\cdot\lb\cdot\|y-x\|$.
So we can set $K_f=(1+\w)\sqrt{m}\cdot\lb$.
\end{proofof}

\begin{proofof}{Lemma~\ref{mnprops}}
Part (i) follows by applying Claim~\ref{glip} to each norm $f_r$, and since the Lipschitz
constant of the maximum of a collection of functions is bounded by the maximum of the
Lipschitz constants of the functions in the collection.
Let $p_{\max}=\max_{i,j}p_{ij}$.
By Claim~\ref{glip}, for each $r\in[k]$, and $S\sse J$ with $|S|=m$, 
both $f_r(\vL[x])/T_r$ and $f_r(\vP[x]_S)/T_r$ have Lipschitz constant at most 
$\sqrt{mn}\cdot p_{\max}\cdot K_r/T_r\leq(1+\w)m\sqrt{n}p_{\max}$. Hence, the Lipschitz
constant of $\mnp$ is at most $K=(1+\w)m\sqrt{n}p_{\max}$.

For part (ii), we mimic the proof of part~\ref{gfo} of Lemma~\ref{gprops}.
Fix $x\in\R^{[m]\times J}$. 
Let $S^*$ be the set of $m$ jobs with the highest $P_j(x)$ values,  
where $P_j(x):=\sum_i p_{ij}x_{ij}$.
Let $\vL=\vL[x]$ and $\vec{P}_{S^*}=\vec{P(x)}_{S^*}$.
Then, 
\[
\mnp(x)=\max\Bigl\{\max_{r\in[k]}f_r(\vec{L})/T_r,\ \ \max_{r\in[k]}f_r(\vec{P}_{S^*})/T_r\Bigr\}.
\]
As in the proof of Lemma~\ref{gprops}~\ref{gfo}, for each $r\in[k]$, we can use $\alg_r$
to obtain an $\w$-first-order oracle for $f_r\bigl(\vL[x]\bigr)/T_r$ and 
$f_r\bigl(\vP[x]_{S^*}\bigr)/T_r$.
Thus, by using Claim~\ref{fotrans} (ii), we obtain a $2\w$-first-order oracle for $\mnp$
using $O(1)$ calls to $\alg_r$, for each $r\in[k]$.
\end{proofof}

\end{document}